\documentclass[a4paper]{article}

\usepackage[l2tabu, orthodox]{nag}
\usepackage[utf8]{inputenc}
\usepackage[pdfusetitle]{hyperref}
\usepackage{microtype}

\usepackage{mathtools}

\mathtoolsset{
showonlyrefs=true 
}

\usepackage{amsmath,amssymb,amsthm}

\newtheorem{lemma}{Lemma}
\newtheorem{theorem}{Theorem}

\usepackage{fullpage}

\usepackage{enumerate}

\usepackage{xspace}

\newcommand{\cV}{\ensuremath{\mathcal V}\xspace}
\newcommand{\cF}{\ensuremath{\mathcal F}\xspace}

\usepackage[usenames,dvipsnames]{color}

\newcommand{\WSP}{the WSP}
\newcommand{\FPT}{\ensuremath{\mathsf{FPT}}\xspace}

\begin{document}

\title{Tight Lower Bounds for the Workflow Satisfiability Problem Based on the Strong Exponential Time Hypothesis}

\author{Gregory Gutin and Magnus Wahlstr{\"o}m 
\\Royal Holloway, University of London, UK\\
}
 \date{}
\maketitle   

\begin{abstract}       
The Workflow Satisfiability Problem (WSP) asks whether there exists an assignment of authorized users to the steps in a workflow specification, subject to certain constraints on the assignment. The problem is NP-hard even when restricted to just not equals constraints.
Since the number of steps $k$ is relatively small in practice, Wang and Li (2010) introduced a parametrisation of WSP by $k$. Wang and Li (2010) showed that, in general, the WSP is W[1]-hard, i.e., it is unlikely that there exists a fixed-parameter tractable (FPT) algorithm for solving the WSP. Crampton et al. (2013) and Cohen et al. (2014) designed FPT algorithms of running time $O^*(2^{k})$ and $O^*(2^{k\log_2 k})$ for the WSP with so-called regular and user-independent constraints, respectively. In this note, we show that there are no algorithms of running time $O^*(2^{ck})$ and $O^*(2^{ck\log_2 k})$ for the two restrictions of WSP, respectively, with any $c<1$, unless the Strong Exponential Time Hypothesis fails. 
\end{abstract} 

\section{Introduction}

The Workflow Satisfiability Problem (WSP) is a problem studied in the security research community, with important applications to information access control. In a WSP instance, one is given a set of $k$ \emph{steps} and a set of $n$ \emph{users}, and the goal is to find an assignment from the steps to the users, subject to some instance-specific \emph{constraints} and \emph{authorization lists}; see formal definition below. 
In practice, the number of steps tends to be much smaller than the number of users. Hence it is natural to study the problem from the perspective of \emph{parameterized complexity}, taking $k$ as a problem parameter.
In general, the resulting parameterized problem is W[1]-hard~\cite{WangLi10}, hence unlikely to be FPT, but for some natural types of constraints the problem has been shown to be FPT.
In particular, Crampton \emph{et al.}~\cite{CrGuYeJournal} gave an algorithm with a running time of~$O^*(2^k)$ for so-called \emph{regular} constraints, and Cohen \emph{et al.}~\cite{CoCrGaGuJo14} gave an algorithm with a running time of $O^*(2^{k \log k})$\footnote{All logarithms in this paper are of base 2.} for \emph{user-independent} constraints; see below. User-independent constraints in particular are common in the practice of access control. 
It was also shown that assuming the Exponential Time Hypothesis (ETH)~\cite{ImPaZa01}, these algorithms cannot be improved to running times of $O(2^{o(k)})$ or $O(2^{o(k \log k)})$, respectively~\cite{CrGuYeJournal,CoCrGaGuJo14}. 
Still, because of the importance of the problem, the question of moderately improved running times, e.g., algorithms of running time $O(2^{ck})$, respectively $O(2^{c k \log k})$, for some $c<1$, remained open and relevant. 
In this paper, we will show that no such algorithms are possible, unless the so-called Strong Exponential Time Hypothesis (SETH)~\cite{ImPa01} fails -- that is, up to lower-order terms, the algorithms cited above are time optimal.

In the remainder of this section, we formally introduce the Workflow Satisfiability Problem (WSP) and some families of constraints of interest for \WSP. We briefly overview \WSP{} literature that considers \WSP{} as a parameterized problem, as  suggested by Wang and Li  \cite{WangLi10}, and state our main results. We prove the results in the next section.

\paragraph{\bf WSP.} In \WSP, the aim is to assign authorized users to the steps in a workflow specification, subject to some constraints arising from business rules and practices.
The  Workflow Satisfiability Problem has applications in information access control (e.g.\ see~\cite{ansi-rbac04,BaBuKa14,BeFeAt99}), and it is extensively studied in the security research community~(e.g.\ see~\cite{BaBuKa14,Cr05,Santos2015,WangLi10}).  
In \WSP, we are given a set $U$ of \emph{users}, a set $S$ of \emph{steps}, a set $\mathcal{A} = \{A(s) :\ s \in S\}$ of \emph{authorization lists}, where $A(s) \subseteq U$ denotes the set of users who are authorized to perform step $s$, and a set $C$ of \emph{constraints}.
In general, a \emph{constraint} $c \in C$ can be described as a pair $c = (T, \Theta)$, where $T \subseteq S$ is the \emph{scope} of the constraint and $\Theta$ is a set of functions from $T$ to $U$ which specifies those assignments of steps in $T$ to users in $U$ that satisfy the constraint (authorizations disregarded).  Authorizations and constraints described in \WSP{} literature are relatively simple such that we may assume that all authorisations and constraints can be checked in polynomial time (in $n=|U|$, $k=|S|$ and $m=|C|$).
Given a \emph{workflow} $W = (S, U, \mathcal{A}, C)$, $W$ is \emph{satisfiable} if there exists a function  $\pi: S \rightarrow U$ called a {\em plan} such that
\begin{itemize}
\item $\pi$ is {\em authorized}, i.e., for all $s \in S$, $\pi(s) \in A(s)$ (each step is allocated to an authorized user);
\item $\pi$ is {\em eligible}, i.e., for all $(T,\Theta) \in C$, $\pi|_{T}  \in \Theta$ (every constraint is satisfied).
\end{itemize}

Wang and Li  \cite{WangLi10} were the first to observe that the number $k$ of steps is often quite small and so can be considered as a parameter. As a result, \WSP{} can be studied as a parameterized problem. 
Wang and Li  \cite{WangLi10}  proved that \WSP{} is {\em fixed-parameter tractable} (\FPT)  if it includes only some special types of practical constraints (authorizations can be {\em arbitrary} as in all other research on WSP mentioned below). This means that \WSP{} restricted to the types of constraints in \cite{WangLi10}
can be solved by an {\em \FPT{} algorithm}, i.e., an algorithm of 
running time $O(f(k)(n+k+c)^{O(1)})=O^*(f(k)),$ where $f(k)$ is a computable function of $k$ only and $O^*$ hides polynomial factors. 
However, in general, \WSP{} is $\mathsf{W[1]}$-hard \cite{WangLi10}, which means that it is highly unlikely that, in general, \WSP{} is FPT.
\footnote{For recent excellent introductions to fixed-parameter algorithms and complexity, see, e.g., \cite{FPTbook,DowneyF13}.}  The paper of Wang and Li  has triggered 
an extensive study of FPT algorithms for \WSP{} from both theoretical and algorithm engineering points of view. We will briefly overview literature on the topic after introduction of some important families of \WSP{} constraints. In what follows, for a positive integer $p$, $[p]$ denotes the set $\{1,2,\dots , p\}$.

\paragraph{\bf WSP Constraints.} We now introduce three families of WSP constraints which consecutively extend each other. 
Let $T$ be a subset of $S$. A plan $\pi$ satisfies a \emph{steps-per-user counting constraint} $(t_\ell,t_r,T)$, if a user performs either no steps in $T$ or between $t_\ell$ and $t_r$ steps. Steps-per-user counting constraints generalize the cardinality constraints which have been widely adopted by \WSP{} community~\cite{ansi-rbac04,BeboFe01,JoBeLaGh05,SaCoFeYo96}. 

For $T\subseteq S$ and $u\in U$ let $\pi\colon\ T\to u$ denote the plan that assigns every step of $T$ to $u$. A constraint $c=(L,\Theta)$ is \emph{regular} if it satisfies the following condition: For any partition $L_1,\ldots ,L_p$ of $L$ such that for every $i \in [p]$ there exists an eligible\footnote{We consider only constraints whose scope is a subset of $L$.} plan $\pi\colon L \rightarrow U$ and user $u$ such that $\pi^{-1}(u) = L_i$, the plan $\bigcup_{i =1}^p (L_i \to u_i)$, where all $u_i$'s are distinct, is eligible. Consider, as an example, a steps-per-user counting constraint $(t_\ell,t_r,L)$. Let  $L_1,\ldots ,L_p$ be a partition of $L$ such that for every $i \in [p]$ there exists an eligible plan $\pi\colon L_i \rightarrow U$ and user $u$ such that $\pi^{-1}(u) = L_i$. Observe that for each $i\in [p],$ we have $t_{\ell} \leq |L_i|\leq t_r$ and so the plan $\bigcup_{i =1}^p (L_i \to u_i)$, where all $u_i$'s are distinct, is eligible. Thus, any steps-per-user counting constraint $(t_\ell,t_r,L)$ is regular.

A constraint $(L,\Theta)$ is {\em user-independent} if whenever $\theta \in \Theta$  and $\psi \colon U \to U$ is a permutation then $\psi \circ \theta \in \Theta$.
In other words, user-independent constraints do not distinguish between users. Observe that all regular constraints are user-independent; however some user-independent constraints are not regular \cite{CrGuYeJournal}.

\paragraph{\bf FPT Algorithms for the WSP.} 
Crampton {\em et al.} \cite{CrGuYeJournal} found a faster \FPT{} algorithm, of running time $O^*(2^k)$, to solve the special cases of WSP studied by Wang and Li \cite{WangLi10} and 
showed that the algorithm can be used for all regular constraints (all constraints studied in  \cite{WangLi10} are regular). Cohen {\em et al.}~\cite{CoCrGaGuJo14} showed that \WSP{} with only user-independent constraints is \FPT{} and can be solved by an algorithm of running time $O^*(2^{k\log k}).$ 
A simpler $O^*(2^{k\log k})$-time algorithm was designed by Karapetyan {\em et al.} \cite{KaGaGu15} for WSP with user-independent constraints. Also an $O^*(2^{k\log k})$-time algorithm was obtained by Crampton {\em et al.}  \cite{CrGuKa15} for a natural optimization version of WSP, the Valued WSP, with (valued) user-independent constraints. The algorithms of these three papers were implemented in \cite{CoCrGaGuJo15,KaGaGu15,CrGuKa15}, respectively, and, in computational experiments, the implementations demonstrated a clear superiority of the FPT algorithms over well-known off-the-shelf solvers, the pseudo-boolean SAT solver SAT4J and the MIP solver CPLEX, for hard WSP and Valued WSP instances (in particular, the off-the-shelf solvers could not find solutions to many instances for which the FPT algorithm found solution within a few minutes). 

Crampton {\em et al.} \cite{CrGuYeJournal} and Cohen {\em et al.}~\cite{CoCrGaGuJo14}, respectively, showed that under the Exponential Time Hypothesis (ETH) \cite{ImPaZa01}, there are no algorithms of running time $O^*(2^{o(k)})$ and $O^*(2^{o(k\log k)})$, respectively, for \WSP{} with regular and user-independent constraints, respectively. However, these results leave possibility of the existence of algorithms of running time $O^*(2^{ck})$ and $O^*(2^{ck\log k})$, respectively, with $c<1$. Such algorithms would not only be of purely theoretical interest, at least in the case of user-independent constraints. The aim of this note is to show that, unfortunately, such algorithms do not exist unless 
the Strong Exponential-Time Hypothesis (SETH) fails. Recall that SETH \cite{ImPa01} states that
\[
\lim_{t \to \infty} \inf \{c \geq 0:\ \textrm{$t$-SAT has an algorithm in time } O(2^{cn})\} = 1.
\]

SETH is a stronger hypothesis than ETH, and has been used repeatedly to argue that various algorithms are ``probably optimal''~\cite{SETH9,LokshtanovMS11a,AbboudW14}.
In this sense, we show that the above-mentioned algorithms for regular respectively user-independent WSP are probably optimal, i.e., that they cannot be improved by current state of the art techniques and that improving them is as hard as improving the running time of SAT algorithms.

\section{Lower Bounds}

It is easy to prove that \WSP{} with regular constraints cannot be solved in time $O^*(2^{ck})$ for any $c<1$ unless SETH fails via a simple reduction from {\sc Set Splitting}. 
In {\sc Set Splitting}, we are given a set $S$ and a family $\{S_1,\dots ,S_p\}$ of its subsets, and our aim is to decide whether the there is a function $f:\ S \rightarrow \{1,2\}$ such that both $f^{-1}(1)\cap S_i$ and $f^{-1}(2)\cap S_i$ are nonempty for every $i\in [p].$ Cygan {\em et al.} \cite{SETH9} proved that {\sc Set Splitting} cannot be solved in time $O^*(2^{c|S|})$ for any $c<1$, unless SETH fails.
To reduce {\sc Set Splitting} to the WSP with regular constraints, let $S$ be the set of WSP steps, $U=\{1,2\}$, $A(s)=U$  for each $s\in S$, and $C=\{(1,|S_i|-1,S_i):\ i\in [p]\}.$  It remains to recall that $(1,|S_i|-1,S_i)$ is a steps-per-user counting constraint, which is regular.

In the rest of this section, we prove that \WSP{} with user-independent constraints cannot be solved in time $O^*(k^{ck})$ for any $c<1$ unless SETH fails. We will show it by an appropriate reduction from $r$-SAT to \WSP{} with user-independent constraints via $(d,r)$-CSP, the Constraint Satisfaction Problem with domain size $d$ and every constraint of arity at most $r$. In $(d,r)$-CSP, we will consider only {\em clause-like constraints}, which are constraints with only one forbidden assignment for the scope variables.\footnote{Note that clauses of CNF SAT are clause-like constraints, which "justifies" the term clause-like. Clearly, an arbitrary CSP constraint can be decomposed into clause-like constraints.}

Let us fix a constant arity $r$, and let $\cF$ be an $r$-SAT formula with $n$ variables. Let us fix a function $f(n)\in O(n^{o(1)}) \cap \omega(\log n)$ such that $n/f(n)$ is a power of 2, e.g., $\frac{1}{2}\log n \log \log n\le f(n)\le \log n \log \log n$.
Let $d=n/f(n)$. 
We will first convert $\cF$ to an instance of $(d,r)$-CSP with $\lceil n/\log d\rceil$ variables, then 
reduce this instance to a WSP instance with appropriate size parameters. 
The following is our first step
(which is simply done by grouping variables).

\begin{lemma}
  There is a reduction from \textsc{$r$-SAT} with $n$ variables
  to \textsc{$(d,r)$-CSP} with only clause-like constraints and with $k=\lceil n/\log d \rceil $ variables, where $d= n/f(n)$.
  The reduction runs in polynomial time.
\end{lemma}
\begin{proof}
  Let the variables of $\cF$ be $X=\{x_1, \ldots, x_n\}$ and $\ell = \log d  = O(\log n)$. 
  For simplicity, add extra variables to $\cF$ so that $n$ is a multiple of $\ell$. 
  Note that this requires adding at most $\ell=o(n)$ new variables.
  We group $X$ into $k=n/\ell$ variable groups $\cV=\{V_1, \ldots, V_k\}$ of $\ell$ variables per group. 
  We also define a new domain $D=\{0,1\}^\ell$. For a variable group $V_i$ and a tuple $b=(b_1,\ldots,b_\ell) \in D$,
  the statement $V_i=b$ is interpreted as the assignment where the $j$'th member of $V_i$ gets value $b_j$. 
  Hence assignments $\cV \to D$ are in 1-1 relationship with assignments $X \to \{0,1\}$. 
  
  Next, for every clause $C$ in $\cF$, we proceed as follows. Let $V(C)$ be the scope of $C$, and
  observe that $C$ is falsified by exactly one assignment to $V(C)$. Similarly, the problem
  $(d,r)$-CSP allows us to arbitrarily specify forbidden assignments to sets of up to $r$ variables.
  Clearly, the variables of $V(C)$ occur in at most $r$ variable groups in $\cV$. We can simply
  enumerate all assignments to these variable groups, these being at most $|D|^r = d^r \leq n^r$
  (recall that $r$ is a constant), and for every such assignment that is an extension of the 
  assignment forbidden by $C$, we add a constraint to $(d,r)$-CSP forbidding this assignment.
  Since this is a polynomial number of constraints for every clause of $\cF$, 
  this can be done in polynomial time in total. (Some of the resulting constraints may have arity
  less than $r$, e.g., some constraints may even be unary. This is allowed in our problem model.)
\end{proof}

Next, we show how we reduce from \textsc{$(d,r)$-CSP} to WSP with user-independent constraints. 

\begin{lemma}
  Let $d=n/f(n)$. 
  There is a polynomial-time reduction from \textsc{$(d, r)$-CSP} with only clause-like constraints and with $k=\lceil n/\log d \rceil$ variables
  to the user-independent WSP with $d$ users and $k+d$ steps. 
\end{lemma}
\begin{proof}
  Consider a $(d,r)$-\textsc{CSP} instance with only clause-like constraints; we will use notation as above,
  i.e., variable set $\cV=\{V_1,\ldots,V_k\}$ and domain $D$. However, number and rename elements of $D$ such that $D=\{1,\ldots,d\}$.
  We create a WSP instance with two sets of steps, and users $U=D$. 
  The \emph{fixed steps} are $S_D=\{s_1, \ldots, s_d\}$,  
  where for $i \in [d]$ the authorization list of $s_i$ is $A(s_i)=\{i\}$. 
  The \emph{free steps} are $S_X=\{s'_1, \ldots, s'_k\}$, 
  each of which has a full authorization list $A(s)=U$. 
  
  Recall that every constraint in the $(d,r)$-\textsc{CSP} instance has a single forbidden assignment.
  Next, for every constraint in the $(d,r)$-\textsc{CSP} instance, over a scope $C=\{V_{q(1)},\dots ,V_{q(p)}\}$ ($p\leq r$)
  with a single forbidden assignment $\phi: C \to D$, 
  we add the following constraint to \WSP{} instance:
  $$\neg(\bigwedge^{p}_{i=1} (s'_{q(i)}=s_{\phi(V_{q(i)})})),$$
  where $s'_{t}=s_j$ means that $s'_{t}$ and $s_j$ must be assigned to the same user.
  Note that the above WSP constraints are user-independent as they do not distinguish between users. 
  It is clear that the reduction can be performed in 
  polynomial time in the size of the input. 

  For correctness, we make two observations. First, by construction, for
  every user $i \in U$ and every authorized plan $\phi: S_D \cup S_X \to U$,
  there is exactly one step $s \in S_D$ such that $\phi(s)=i$. 
  Hence, for every step $s' \in S_X$, we have $\phi(s')=i$ if and 
  only if $\phi(s')=\phi(s_i)$. Second, let $\phi$ be an authorized
  plan as above, and define $\phi': \cV \to D$ by $\phi'(V_i)=\phi(s'_i)$. 
  Then (by the previous observation) for every constraint $C$ of the $(d,r)$-\textsc{CSP} instance,
  $\phi'$ satisfies $C$ if and only if $\phi$ is eligible for the corresponding
  constraint of \WSP{} instance. This shows the equivalence of the instances.
\end{proof}

Note that our WSP instance has constraint of bounded arity (at most $2r$). 
This may not be critical, but it alleviates some potential concerns 
(e.g., the specific encoding of the WSP constraints is not important).
We can now wrap up the proof.

\begin{theorem}
  \WSP{} with user-independent constraints cannot be solved in time
  $O^*(2^{ck \log k})$ for any $c<1$ unless SETH fails.
\end{theorem}
\begin{proof}
  In this proof, for functions $g'(n)$ and $g''(n)$, we write $g'(n)\sim g''(n)$ if $g'(n)=g''(n)(1+o(1))$. Observe that 
  $g'(n)\sim g''(n)$ and $g''(n)\sim g'''(n)$ imply  $g'(n)\sim g'''(n).$

  Chaining the two reductions above, we have a polynomial-time reduction
  from an $r$-SAT instance $\cF$ on $n$ variables to \WSP{} instance on $k+d$ variables,
  where $d=n/f(n) $ and $k=\lceil n/\log d\rceil$. In particular, we can write
  \[
  k \sim \frac{n}{\log n - \log f(n)} = \frac{n}{\log n} \cdot \frac {\log n}{\log n-\log f(n)} =
  \frac{n}{\log n}\left(1+\frac{\log f(n)}{\log n - \log f(n)}\right) \sim \frac{n}{\log n}.
  \]
  Similarly, we have $d=\lceil n/f(n)\rceil =o(n/\log n)$, hence 
  \[
  k'=k+d \sim \frac{n}{\log n}.
  \]  
  Now note that
  \[
  k' \log k' = \frac{n}{\log n}(1+o(1))(\log n - \log \log n + o(1)) \sim n.
  \]
  Hence for any $c<1$, solving \WSP{} instance in the stated time would imply
  solving every $r$-SAT instance for every $r$ in time $O(2^{c'n})$ for some $c'<1$ independent of $r$.
  This would contradict SETH.
\end{proof}

\section*{Acknowledgments} Gutin's research was partially supported by EPSRC grant  EP/\linebreak[1]K005162/1 and by Royal Society Wolfson Research Merit Award.


\end{document}